\documentclass[11pt]{article}

\usepackage{amsfonts}
\usepackage{graphicx}
\usepackage{epstopdf}
\usepackage{algorithmic}

\usepackage{amsopn}

\usepackage{graphicx} 
\usepackage{fullpage}  
\usepackage{amsfonts}
\usepackage{amssymb}
\usepackage{mathtools}
\usepackage{enumerate}
\usepackage{caption}  
\usepackage{subcaption}
\usepackage{hyperref}       
\usepackage{url}            

\usepackage{amsthm,amscd}

\usepackage{tikz}
\usetikzlibrary{arrows,positioning,shapes,3d,calc}
\usepackage{pgfplots}
\usepgfplotslibrary{groupplots}

\pgfplotsset{compat = newest,
		plot coordinates/math parser = false}
\newlength\figureheight
\newlength\figurewidth

\DeclareMathOperator*{\argmin}{arg\,min}
\DeclareMathOperator*{\rank}{rank}
\DeclareMathOperator*{\Diag}{Diag}
\DeclareMathOperator*{\Image}{Im}

\DeclareMathOperator{\diag}{diag}

\newtheorem{definition}{Definition}[section]
\newtheorem{theorem}{Theorem}[section]

\newtheorem{corollary}{Corollary}[section]

\newtheorem{remark}{Remark}[section]

\numberwithin{equation}{section}

\begin{document}

\title{Analysis of the Neighborhood Pattern Similarity Measure for the Role Extraction Problem}

\author{Melissa Marchand\footnotemark[1] \and Kyle Gallivan\footnotemark[2]
\and Wen Huang\footnotemark[3]
\and Paul Van Dooren\footnotemark[4]}

\maketitle

\begin{abstract}
 In this paper we analyze an indirect approach, called the  Neighborhood Pattern Similarity approach, to solve the so-called role extraction problem of a large-scale graph. The method is based on the preliminary construction of a node similarity matrix which allows in a second stage to group together, with an appropriate clustering technique, the nodes that are assigned to have the same role. The analysis builds on the notion of ideal graphs where all nodes with the same role, are also structurally equivalent. 	
\end{abstract}


\footnotetext[1]{Department of Mathematics, Florida State University, 208 Love Building, 1017 Academic Way, Tallahassee, FL 32306-4510, USA.
	\url{melissa.s.marchand@gmail.com}.}
\footnotetext[2]{Department of Mathematics, Florida State University, 208 Love Building, 1017 Academic Way, Tallahassee, FL 32306-4510, USA.
	\url{gallivan@math.fsu.edu}.}
\footnotetext[3]{School of Mathematical Sciences, Fujian Provincial Key Laboratory of Mathematical Modeling and High-Performance Scientific Computing,
        Xiamen University, Xiamen, Fujian, P.R.China, 361005. 
        \url{wen.huang@xmu.edu.cn}.}
\footnotetext[4]{Department of Mathematical Engineering, Universit\'e catholique de Louvain, Louvain-La-Neuve, Belgium.
	\url{paul.vandooren@uclouvain.be}.}



\section{Introduction}
To analyze large networks and obtain relevant statistical properties, clustering nodes together into subgroups of densely connected nodes, called communities, is a popular approach. Various measures and algorithms have been developed to identify these community structures \cite{POM:2009, Fortunato:2010, Traag:2014}. However, there are network structures that cannot be determined using community detection algorithms, such as bipartite and cyclic graph structures, which appear in human protein-protein interaction networks \cite{PSR:2010} and food web networks \cite{GSSPLNA:2010}, respectively. General types of network structures are known as role structures, and the process of finding them is called the role extraction problem, or block modeling.

The role extraction problem determines a representation of a network by a smaller structured graph, called the reduced graph, role graph, or image graph, where nodes are grouped together into roles based upon their interactions with nodes in either the same role or different roles. This problem is a generalization of the community detection problem where each node in a community mainly interacts with other nodes within the same community and there are no, or very few, interactions between communities.  There are many real world applications to which role extraction can be applied and from which characterizations of  interactions that define roles can be taken, such as studying trade networks between countries \cite{RW:2007}; evaluating the resilience of peer-to-peer networks \cite{HKYH:2002}; ranking web pages in search engines \cite{PBMW:1999}; studying human interaction by email correspondence \cite{AMC:2010}; modeling protein-protein interactions \cite{KKKR:2002}; and analyzing food webs \cite{GSSPLNA:2010}.

Previous research solved the role extraction problem using either direct or indirect approaches, where direct approaches cluster the network directly into roles~\cite{DBF:2005,RW:2007,Reichardt:2009}, while indirect approaches construct a  node similarity matrix of the data set and then cluster highly similar nodes together~\cite{BGHSVD:2004,LHN:2006,CB:2011,Cason:2012,BDVB:2013,CLVD:2016,Marchand:2016,Marchand:2017}. Both approaches have strengths and weakness for solving the role extraction problem.  A strength of direct approaches is that it explicitly fits the data into a role structure. Unfortunately, there is  no well-accepted measure to determine whether or not a role assignment fits the data, so a priori knowledge about the network is necessary or multiple role assignments must be tested to determine the best role structure for the data~\cite{DBF:2005}.  

Indirect approaches do not require an assumption on the role assignment and may reveal complex network structures that the original data may not reveal. The main problem with indirect approaches is that there exist several different types of node  similarity measures. In addition, many of these measure have been deemed unsuitable for the role extraction problem due to difficulties encountered when extracting role structures from certain types of graphs (e.g., regular graphs and normal graphs), loss of information (e.g., the origin, the destination, and the intermediate nodes involved in the transmission of the flow), or were more suited to detect community structures than role structures~\cite{Browet:2014}. Fortunately, recent work has shown that the neighborhood pattern similarity measure can be used to solve the role extraction problem when using the indirect approach.

Browet and Van Dooren used the neighborhood pattern similarity measure to solve the role extraction problem and showed empirically that the measure was able to determine the role structure of complex networks \cite{Browet:2014}. In addition, they developed an algorithm to compute a low-rank approximation of the similarity matrix  and showed empirically that their indirect approach can extract role structures within networks. Marchand improved upon their low-rank algorithm using Riemannian optimization techniques to develop a more efficient algorithm to compute the low-rank similarity matrix and showed (analytically and empirically) that there exists a relationship between the rank of the similarity matrix and the number of roles in the network \cite{Marchand:2017}.

In this paper, we analyze the neighborhood pattern similarity measure and show that, under certain assumptions, we can recover roles from a low-rank factorization of the similarity matrix due to the relationship between the rank of the similarity matrix and the number of roles. Also, we explore how perturbing the adjacency matrix affects the singular values (and rank) of the similarity matrix. Lastly, we unify special complex structures in networks (e.g., community, overlapping community, etc.) as role structures and show that the neighborhood pattern similarity measure can be used as well to find these structures in network topology.

\section{Role Extraction Problem}
Given a (un)weighted and directed network, the role extraction problem represents the network by its adjacency matrix  and determines a representative role structure for the network. This role is determined by assuming that nodes can be grouped according to a suitable measure of equivalence. In this section, we state two measures of equivalence used for the role extraction problem and define the general form and state the constraints necessary to extract viable role structures. Most of the discussion and results that follow concern unweighted directed graphs. In Section~\ref{sec:complex}, the expression of special graph structures as role structures includes generalizations to signed weighted directed graphs.

\subsection{Measure of Equivalence and Definition of the Role Extraction Problem}
A graph, denoted $G(V, E)$, is a mathematical structure with two finite sets $V$ and $E$, where the elements of the set $V = \{1, \ldots, n\}$ are called nodes and the elements of the set $E=\{(i, j)~\vert~i, j\in V\}$ are called edges. If there exists an edge between nodes $i$ and $j$, i.e., the pair $(i, j) \in E$, then nodes $i$ and $j$ are adjacent. The adjacency matrix is an $n\times n$ $\{0,1\}$-matrix $A$, where if $(i,j) \in E$, then $A_{i,j} = 1$; otherwise $A_{i,j} = 0$. If the graph is weighted and $(i,j) \in E$, then the weighted adjacency matrix is denoted by $W$ and $W_{i,j}$ is represented by its edge weight.

Given the adjacency matrix $A$, the role extraction problem finds a $n \times n$ permutation matrix $P$ such that the edges in the permuted adjacency matrix $A_{p} := P^{T}AP$, which represents the relabeled graph, are mainly concentrated into blocks (see Figure~\ref{fig:RMPexample}). In order to form the relabeled graph, one needs to determine if the nodes are structurally or regularly equivalent.

\begin{figure}[!h]\centering
\includegraphics[scale=1.0]{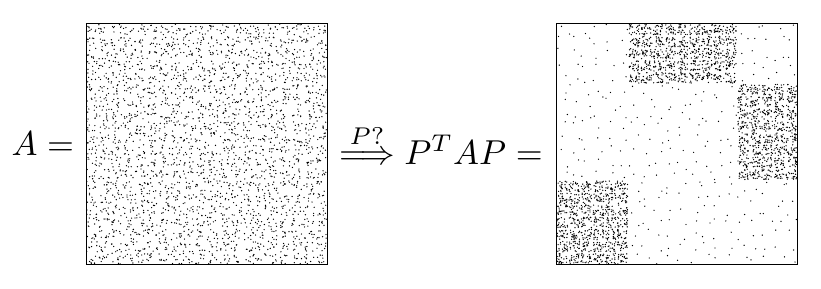}
\caption{Block modeling : find the permutation $P$ such that the relabeled adjacency matrix $P^TAP$ 
has an approximate block structure. \label{fig:RMPexample}}
\end{figure}

Two nodes are structurally equivalent if they have exactly the same children and the same parents \cite{LW:1971}. In terms of block modeling, this means that all blocks in the permuted adjacency matrix must then either have only 0's or only 1's~: 
if we denote by $I$ and $J$ two sets of nodes that each contain structurally equivalent nodes, then
(i)  the block $(I,J)$ in $P^TAP$ is a ``0'' block if {\em none} of the nodes in $I$ has {\em any} children in group $J$, or equivalently, 
 if {\em none} of the nodes in $J$ has {\em any} parents in group $I$, and (ii) it is a ``1'' block if {\em all} the nodes in group $I$ has {\em all} the nodes in group $J$ as children, or equivalently,  if {\em all} the nodes in group $J$ has {\em all} the nodes in group $I$ as parents \cite{DBF:2005}.

Structural equivalence usually extracts many small roles in networks \cite{WR:1983,EB:1994,EB:1996}. Thus, an alternative equivalence relation, regular equivalence, was proposed to extract larger roles. Two nodes are regularly equivalent if, while they do not necessarily share the same neighbors, they have neighbors who are themselves structurally or regularly equivalent.  Alternatively, this means that the blocks in the permuted adjacency matrix must contain at least one element per row and column (called a regular block). Note that structural equivalence implies regular equivalence, but regular equivalence does not imply structural equivalence.

Every group of regularly equivalent nodes of $A$ can be represented by a single role in the role graph, which is a smaller assignment matrix $B$ with a number of nodes that is the number of groups of regular equivalent nodes in $A$. Moreover, $B_{I,J}=1$ if the nodes in group $I$ in the original graph $A$ all point to the nodes in group $J$ in $A$, and $B_{I,J}=0$
if none of the nodes in group $I$ in $A$ point to any of the nodes in group $J$ in $A$. Additionally, Reichardt and White assumed that no two roles in the role graph $B$ may be structurally equivalent  because if they were, then both roles would interact with the same roles and we would be unable to distinguish between the two roles \cite{RW:2007,Reichardt:2009}. So, these two roles should be merged into one role. These ideas are used below to define an ideal form of adjacency matrix that facilitates the extraction of roles by its use to approximate the adjacency matrix of the given graph.

Earlier research in role extraction involved creating a cost function to minimize over both the role structure and role assignment of nodes in the graph based on a choice of equivalence relation~\cite{WF:1994,DBF:2005, RW:2007,Reichardt:2009}. That is, if $B$ is the adjacency matrix of the role graph  and $\sigma$ is the assignment of each node to a role, then the problem can be stated as
\begin{equation}\label{eq:rmp}
(B^{*}, \sigma^{*}) = \argmin_{B, \sigma} Q_{A}(B, \sigma),
\end{equation}
where $Q_{A}$ depends on the graph topology and chosen equivalence criterion. Note that~\eqref{eq:rmp} is a combinatorial optimization problem with respect to two groups of variables and is harder than the community detection problem, which is, in general, NP-hard~\cite{BRGGHNW:2006,BRGGHNW:2008}. The cost function $Q_{A}(B, \sigma)$ can either be constructed indirectly, based on a (dis)similarity measure between pairs of nodes, or directly, based on a measuring of the fit of clusters compared to an ideal clustering with perfect relations within and between clusters. We focus on an indirect approach to the role extraction problem and show how the similarity metric chosen for our approach can be used to first extract the optimal assignment function $\sigma^*$ of the role structures in a network. Once the groups of the assignment $\sigma^*$ have been identified, then in a second step, the role matrix $B$ is easy to construct, provided we use a cost function that is ``decoupled'' in the elements of $B$,  since each element $B_{ij}\in \{0,1\}$ can then be chosen independently in order to maximize $Q_A(B,\sigma^*)$.

\subsection{Role Models and Ideal Graphs}

We are particularly interested in graphs with a permuted adjacency matrix $A_P:=P^TAP$ which has a special block form that can be represented in the factorized form
\begin{equation} \label{struct} 
A_P=ZBZ^T, \quad  Z = \begin{bmatrix} z_{1} & 0 & \cdots & 0 \\
0 & z_{2} & \cdots & 0 \\ \vdots & \vdots & \ddots & \vdots \\ 0 & 0 & \cdots & z_{q} \end{bmatrix} =  \Diag\{ z_{1}, \ldots, z_q\} \in \mathbb{R}^{n\times q}, 
\end{equation}
where $q$ is the number of roles in the role graph, $z_i:=[1, \ldots, 1]^T\in \mathbb{R}^{n_i}$, $n = n_{1} + \cdots + n_{q}$, and $B$ is a $q\times q$ adjacency matrix (i.e., the role matrix) describing the roles in the original matrix $A_P$. 
We assume that the graph does not have disconnected nodes (a zero row and corresponding zero column of the adjacency matrix $A$) because this would imply that $Z$ has a zero row. We call such graphs \textbf{ideal graphs} because all of the nodes in each role are structurally equivalent. 

Such a decomposition is not unique. If, for instance, the $k$-th row and column of the matrix $B$ contains only zeros , then clearly the row and column can be removed  in the decomposition $ZBZ^T$, which yields a smaller decomposition.  For example, suppose the matrix $B$ and the decomposition matrices are
\begin{equation*}
A_P=ZBZ^T, \quad B:= \left[\begin{array}{ccc} 0 & 1 & 0 \\ 1 & 1 & 0 \\ 0 & 0 & 0  \end{array}\right], \quad  Z := \Diag\{ z_1, z_2 , z_3\}
\end{equation*}
and the matrix  $A_P$ can be represented by the 2 role decomposition $A_P=\hat Z \hat B \hat Z^T$, where 
\begin{equation*}
A_P=\hat Z\hat B\hat Z^T, \quad \hat B:= \left[\begin{array}{cc} 0 & 1  \\ 1 & 1   \end{array}\right], \quad \hat Z := \Diag\{ z_1, z_2 \}.
\end{equation*}
This implies that the original adjacency matrix also had a number of corresponding zero rows and columns. Therefore, we do not need to associate any role to the corresponding nodes.
 
Non-uniqueness of the factorization also occurs when the matrix $B$ has itself a decomposition with a smaller matrix $\hat B$ with fewer roles. Consider the $3\times 3$ image matrix in the decomposition
\begin{equation*}
A_P=ZBZ^T, \quad B:= \left[\begin{array}{ccc} 1 & 1 & 1 \\ 1 & 0 & 0 \\ 1 & 0 & 0  \end{array}\right], \quad  Z := \Diag\{ z_1, z_2 , z_3\}.
\end{equation*} 
 The roles 2 and 3 of the image matrix $B$ are structurally equivalent and can be combined into a single role. This implies that $B$ has the factorization as
\begin{equation*}
\left[\begin{array}{ccc} 1 & 1 & 1 \\ 1 & 0 & 0 \\ 1 & 0 & 0  \end{array}\right]=
 \left[\begin{array}{cc} 1 & 0 \\  0 & 1 \\ 0 & 1  \end{array}\right] \left[\begin{array}{cc} 1 & 1 \\  1 & 0 \end{array}\right]
  \left[\begin{array}{ccc} 1 & 0 & 0\\  0 &  1 & 1  \end{array}\right]
\end{equation*}
 which can then be used to obtain the smaller decomposition
\begin{equation*}
A_P=\hat Z\hat B\hat Z^T, \quad \hat B:= \left[\begin{array}{cc} 1 & 1  \\ 1 & 0  \end{array}\right], \quad \hat Z := \Diag\{ z_1, \hat z_2 \}
\end{equation*}
where $\hat z_2$ has now $n_2+n_3$ elements.
 
In order to introduce a form of uniqueness, we define the so-called {\em minimal} role matrices.

\begin{definition}
 \label{th:GS2} If $A=ZBZ^T$ is an adjacency matrix of a connected unweighted directed ideal graph then $B$ is a \textbf{minimal role} matrix if no two rows 
of the compound matrix  $\left[\begin{array}{cc} B & B^T  \end{array}\right]$ are linear dependent.
\end{definition}

Note that, unlike community detection,  a factorization $A=ZBZ^T$ always exists by simply taking $A=B$ and $Z=I$. In practice, this is of little interest since the point of role extraction is to identify structure in the graph with significantly fewer roles than nodes, i.e., a low-rank ideal adjacency matrix that approximates $A$ well. As a result, we explore the relationship between the neighborhood pattern similarity measure, minimal ideal graphs, the rank of their adjacency matrices, the number of roles, and the rank of $A$.

\subsection{Uniform distributions}
Browet et al. showed empirically that the indirect method using the neighborhood pattern similarity measure worked well extracting the role structure from randomly generated Erd\"{o}s-R\'{e}nyi graphs~\cite{Browet:2014}. The method works well on randomly generated Erd\"{o}s-R\'{e}nyi graphs because their expected value is a matrix of rank 1.  

In an Erd\"{o}s-R\'{e}nyi graph, each node has a probability $p$ to be present. Therefore the expected value of the adjacency matrix of such an $n\times n$ graph equals $E(A)=p\mathbf{1}\mathbf{1}^T$,
which is rank 1 and has Perron root $np$. Moreover, it has been shown that for large Erd\"{o}s-R\'{e}nyi  graphs (these are undirected graphs with edges of equal probability $p$) the other eigenvalues of the adjacency matrix have an expected value that is  much smaller than $np$~\cite{EKYY:2013}. Therefore, if we apply the same reasoning to a matrix $A$ with role matrix $B$, then the expected value of the adjacency matrix would be
\begin{equation*}
E(A) = (PZ) [p_{in}B + p_{out}(\mathbf{1}\mathbf{1}^T-B)](PZ)^T = (PZ)E(B)(PZ)^T,
\end{equation*}
where $p_{in}$ is the probability of an edge existing between corresponding roles and $p_{out}$ is the probability that an edge does not exist. Therefore, $E(A)$ has rank at most $q$ and if the remaining $n-q$ eigenvalues are also small such as in the standard Erd\"os-Renyi case, then $E(A)$ is a good approximation of $A$.

\section{Analysis of the Rank of the Neighborhood Pattern Similarity Measure}
In practice, computing the similarity matrix is expensive, especially for large networks. So, a low-rank approximation of the similarity matrix is preferable due to efficiency in storage and computational complexity, but  we may lose information necessary to extract the role structure. In this section, we prove for, the ideal graph case, that there exists a relationship between the rank of neighborhood pattern similarity matrix and the number of roles in the network. In addition, we prove that the roles can be extracted correctly, even for a similarity matrix with rank less than the number of roles. Section~\ref{sec:perturb} considers extracting the role structure of graphs whose similarity measure have a good low-rank approximation by examining the neighborhood of graphs around ideal graphs with low-rank similarity measures.

\subsection{Neighborhood Pattern Similarity Measure}
The neighborhood pattern similarity measure determines if two nodes are similar if they have similar neighborhood patterns \cite{Browet:2014,BVD:2014,Denayer:2012}. A neighborhood pattern of length $\ell$ is defined as the number of incoming (I)  and outgoing (O) edges starting from a source node \cite{Browet:2014}. For example, neighborhood patterns of length $1$ are patterns where two nodes are similar if they have common parents, i.e., Figure~\ref{subfig:neighborhoodPatternsLength1_pattern1}, or common children, i.e., Figure~\ref{subfig:neighborhoodPatternsLength1_pattern2}. The number of common parents between  two nodes $(i, j)$  is the number of nonzero row elements shared by the $i$-th and $j$-th columns of $A$, i.e., $[A^{T}A]_{i, j}$ and the number of common children is the number of nonzero column elements shared by the $i$-th and $j$-th rows of $A$, i.e., $[AA^{T}]_{i, j}$. Therefore, the number of common reachable nodes, called target nodes, between every pair of source nodes for neighborhood patterns of length $1$ is $N_{1} = AA^T + A^TA$ \cite{Browet:2014}.

\begin{figure*}[!h]\centering
\begin{subfigure}{0.4\textwidth}\centering
\includegraphics[scale=0.5]{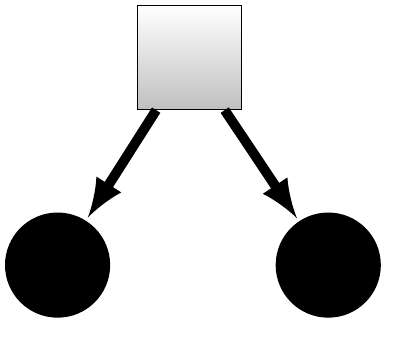}
\caption{Pattern I: $A^TA$ \label{subfig:neighborhoodPatternsLength1_pattern1}}
\end{subfigure}
~
\begin{subfigure}{0.4\textwidth}\centering
\includegraphics[scale=0.5]{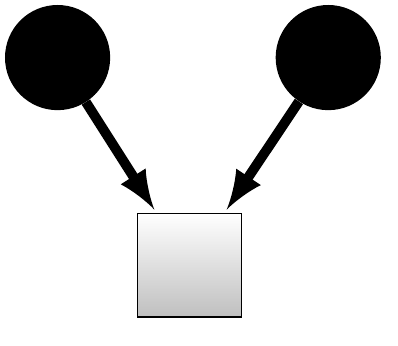}
\caption{Pattern O: $AA^T$ \label{subfig:neighborhoodPatternsLength1_pattern2}}
\end{subfigure}
\caption{All possible neighborhood patterns of length $1$ for the similarity measure where the source nodes $i$, $j$ are the black circles and the target node is the gray square \label{fig:neighborhoodPatternsLength1}}
\end{figure*}

For neighborhood patterns of length $2$, there are four possible neighborhood patterns and the number of common target nodes between every pair of source nodes for neighborhood patterns of length $2$ is given by (see~\cite{Browet:2014} for a more detailed proof of this)
\begin{align}
N_{2} &= AAA^TA^T + AA^TAA^T + A^TAA^TA + A^TA^TAA = AN_{1}A^T + A^T N_{1} A.\nonumber
\end{align}

In general, the number of possible neighborhood patterns of length $\ell$ is $2^{\ell}$ and the number of common target nodes is given by
\begin{equation*}
N_{\ell} = AN_{\ell-1} A^T + A^T N_{\ell-1} A.
\end{equation*}
Therefore, the neighborhood pattern pairwise node similarity measure can be defined as the weighted sum of the number of common target nodes of the neighborhood patterns of any length, i.e.,
\begin{equation}\label{eq:similaritymeasureSumRepresentation}
S = \displaystyle\sum_{\ell = 1}^{\infty} \beta^{2(\ell-1)}N_{\ell},
\end{equation}
where $\beta\in\mathbb{R}$ is a scaling parameter that weights longer neighborhood patterns \cite{Browet:2014} (implying that as $\beta$ increases it is expected to be more difficult to have two nodes similar to each other). Note that the similarity matrix $S$ is a symmetric positive semi-definite matrix.

\subsection{The Similarity Matrix Recurrence}
In Browet's thesis \cite{Browet:2014}, the following recurrence relation was proposed for computing a similarity matrix, where $\Gamma_A[X]:=AXA^T+A^TXA$ is a linear mapping from $\mathbb{R}^{n\times n}$ to  $\mathbb{R}^{n\times n}$
which moreover preserves symmetry, non-negativity and semi-definiteness of the argument $X$~:
\begin{equation} \label{iterate}  
S_1 := \Gamma_{A}[I_{n}] = AA^T+A^TA, \quad S_{k+1}:= \Gamma_A[ I_{n} + \beta^2 S_k], \; \forall k \ge 1,
\end{equation} 
where $I_{n}$ is the $n\times n$ identity matrix. This sequence was shown to converge to a bounded fixed point
\begin{equation}
S_\infty= \Gamma_A[ I_{n} + \beta^2 S_\infty]
\end{equation}
if and only if $\beta^2$ satisfies
\begin{equation}\label{converge}
\beta^2 < \frac{1}{\rho(A\otimes A+A^T\otimes A^T)}.
\end{equation} 
Since the initial matrix $S_1$ is symmetric, positive semi-definite and non-negative, and the mapping $\Gamma_A$ preserves these properties, it follows that all matrices $S_k$ are symmetric, positive semi-definite and non-negative.

\begin{theorem} \label{th:similarity} 
Consider the iteration \eqref{iterate} where $\beta$ is chosen according to \eqref{converge} to guarantee convergence to a bounded solution $S_{\infty}$.
Then all matrices $S_k$, including their limit $S_{\infty}$, have the same image as the compound matrix $[A~A^T]$, and the same rank $r$. 
\end{theorem}
\begin{proof}
The iteration \eqref{iterate} can be rewritten as follows
\begin{equation*}
S_1:= \left[\begin{array}{cc} \! A \! & \! A^T \! \end{array}\right]\left[\begin{array}{cc} \! A^T \!  \\ A \end{array}\right], \; S_{k+1} = \left[\begin{array}{cc} \! A \! & \! A^T \! \end{array}\right]
\left[\begin{array}{cc} \! I+\beta^2S_k \! & 0 \\ 0 & \! I+\beta^2S_k \!   \end{array}\right]\left[\begin{array}{cc} \! A^T \! \\ A \end{array}\right], \;  \forall k \ge 1.
\end{equation*}
The first equation implies  $\Image S_1 = \Image \left[ A ~ A^T \right]$, where $\Image M$ denotes the image (or column space) of a matrix $M$. In the second equation, the middle matrix is positive definite since $S_k$ is semi-definite, and this implies that  $\Image S_{k+1} = \Image \left[ A ~A^T \right]$. This also implies that the rank of all matrices $S_k,~k \geq 1$
is equal to the rank of  $\left[ A ~A^T \right]$. For the limit $S_\infty$, one has to be more careful, since the rank could drop. But the second equation also implies that  (in the Loewner ordering) $S_k\succeq S_1$ for all $k$ and hence the rank must remain constant.
\end{proof}

\subsection{The Ideal Graph Case}
For the role extraction problem, the idea of using a low-rank projection of a similarity measure for the construction of the indirect cost function was proposed by Browet and Van Dooren in~\cite{Browet:2014,BVD:2014}. However, they only provided empirical evidence of the relationship between the rank of the neighborhood pattern similarity measure and the number of roles in  the network. In this section, we prove, for ideal graphs, a relationship between the rank factorization of the similarity matrix and the number of roles in the network. In addition, we prove that the rows of the left factor of the rank factorization of the similarity matrix have exactly $q$ clusters of all parallel vectors.

Recall that adjacency matrix, $A$, of an ideal graph satisfies the decomposition
\begin{equation*}
A = (PZ)B(PZ)^T
\end{equation*}
where $B$ is a minimal role matrix.
\begin{corollary}
When the  matrix $A$ is the adjacency matrix of an ideal graph with minimal role matrix $B$, then $\Image S_k\subseteq \Image PZ$, and hence
$r:=\rank{S_k} \le q := \rank{Z}$, where $q$ is the number of roles.
\end{corollary}
\begin{proof}
This is a direct consequence of Theorem \ref{th:similarity} since both $A$ and $A^T$ have an image included in the image of $PZ$.
\end{proof}

Since $S_k$ is symmetric and  $\Image S_k\subseteq  \Image \! PZ$, we can write it as $S_k\! =\! (PZ)\hat S_K(PZ)^T,$ where $\hat S_k:=N^{-1}(PZ)^TS_k(PZ)N^{-1} \succeq 0$ (i.e., $\hat S_k$ is positive semi-definite) and
$N:=Z^TZ=\diag\{\lVert z_{1} \rVert_{2}^{2}, \ldots, \lVert z_{q} \rVert_{2}^{2}\}$ is a $q\times q$ diagonal matrix. We can then reformulate the iteration \eqref{iterate} as a recurrence for the $\hat S_k$ matrices~: 
\begin{eqnarray*} 
\hat S_1 & :=& \left[\begin{array}{cc} B & B^T \end{array}\right] \left[\begin{array}{cc} N & 0  \\ 0 & N  \end{array}\right] \left[\begin{array}{cc} B^T \\ B \end{array}\right], \\
\hat S_{k} &:=& \left[\begin{array}{cc} B & B^T \end{array}\right]
\left[\begin{array}{cc} N_k & 0 \\ 0 & N_k   \end{array}\right]\left[\begin{array}{cc} B^T \\ B \end{array}\right]
, \;  \forall k > 1
\end{eqnarray*}
where $N_k:=N+\beta^2N\hat S_{k-1}N$.
We then obtain the following result
\begin{corollary}
When $A$ is an adjacency matrix for an ideal graph with minimal role matrix $B$, then for all $k$, including $k=\infty$, it follows that $\Image \hat S_k = \Image [B~B^T]$,  $\Image S_k= \Image PZ\hat S_k = \Image PZ[B~B^T]$, and hence
$r:=\rank{\hat S_k}=\rank{S_k} \le q := \rank{Z}$, where $q$ is the number of roles.
\end{corollary}
\begin{proof}
The proof that  $\Image \hat S_k = \Image [B~B^T]$ is very similar to the proof of Theorem \ref{th:similarity}. The fact that   $\Image S_k= \Image PZ\hat S_k$
follows from the identity $S_k=(PZ)\hat S_k(PZ)^T$ and the fact that both $S_k$ and $\hat S_k$ are semi-definite. The rest easily follows.
\end{proof}

Now, we show that even when we have a factorization \eqref{struct} with the minimal role matrix $B$, $\rank{S_k}$ can be smaller than $q$, the number of roles. Let $A = (PZ)B(PZ)^T$ with 
\begin{equation*}
B:=\left[\begin{array}{ccc} 0 & 0 & 0 \\ 1 & 0 & 1 \\ 1 & 0 & 1  \end{array}\right].
\end{equation*}
Let us for simplicity choose $(PZ)=I_3$, i.e. $A=B$, $\hat S_k=S_k$ for all $k$, and $q=3$. 
It is easy to check then that $\rank{A}=\rank{B}=1$, and
\begin{equation*}
\rank{[A~A^T]}=\rank{[B~B^T]}=\rank{\left[\begin{array}{ccc|ccc} 0 & 0 & 0 & 0 & 1 & 1 \\ 1 & 0 & 1 &  0 & 0 & 0 \\ 1 & 0 & 1 & 0 & 1 & 1 \end{array}\right]}=2,
\end{equation*}
and
\begin{equation*}
\rank S_k = \rank{S_1}=\rank{\left[\begin{array}{ccc} 2 & 0 & 2 \\ 0 & 2 & 2 \\ 2 & 2 & 4  \end{array}\right]}=2, \quad S_1=2\left[\begin{array}{ccc} 1 & 0\\ 0 & 1 \\ 1 & 1 \end{array}\right]\left[\begin{array}{ccc} 1 & 0 & 1 \\ 0 & 1 & 1 \end{array}\right].
\end{equation*}

Nevertheless, we prove below that we can recover the different roles via clustering performed on the low-rank factorization of any matrix ${S}_k$, even though its rank is smaller than the number of roles.

\begin{theorem}\label{thm:lowrankfactor} Let $A=(PZ)B(PZ)^T$ be an ideal graph with $B$ minimal. For any $k$, let $\hat S_k \succeq 0$ be defined as above with rank $r\le q$ and low-rank factorization $\hat S_k = V_kV_k^T$ where $V_k\in\mathbb{R}^{q\times r}$. 
The matrix $S_k$ has the low-rank factorization $S_k=(PZV_k)(PZV_k)^T$ and the rows of the matrix $PZV_k$ have exactly $q$ clusters of all parallel vectors.
\end{theorem}
\begin{proof}
Since all matrices $\hat S_k=\left[\begin{array}{cc} B & B^T\end{array}\right]\diag(N_k,N_k)\left[\begin{array}{cc} B^T \\ B\end{array}\right]$ are symmetric matrices with the same image, they also have the same kernel, which must be the kernel of $\left[\begin{smallmatrix} B^T \\ B\end{smallmatrix}\right]$.
Let $v_{j,k}^T,~j=1,\ldots,q$ be the rows of the matrix $V_k$. Then $v_{j,k}\neq 0$ for any $j$, since otherwise the $j$-th row of $B$ and $B^T$ would be zero and this violates the minimality assumption. For the same reason, no two rows $v_{i,k}^T$ and $v_{j,k}^T$ of $V_k$ can be parallel, since otherwise the same two rows of $B$ and $B^T$ would be parallel, which means that these roles would be structurally equivalent and this also violates the minimality assumption. 
If we now look at the (unpermuted) matrix $ZV_k$ then it has the form
\begin{equation*}
ZV_k = \left[\begin{array}{c} z_1v_{1,k}^T \\ \vdots \\ z_qv_{q,k}^T \end{array}\right],
\end{equation*}
and each block $z_jv_{j,k}^T$ corresponds to a cluster of $n_j$ row vectors parallel to $v_{j,k}^T$. Since none of the vectors $v_{j,k}$ is zero or parallel to 
another row of $V_k$, we have exactly $q$ different clusters. This is of course not affected by the permutation $P$.
\end{proof}

\begin{remark}
If $B$ is not a minimal role matrix, then the correct number of roles of the compressed image matrix $\hat B$ can be detected by removing zero rows, or merging dependent rows, of $V$. \hfill $\Box$
\end{remark}

By Theorem~\ref{thm:lowrankfactor}, we can extract the role structure of the network from the low-rank factor of any $S_k$ of an ideal adjacency matrix. This is ideal for large networks because the similarity matrix recurrence~\eqref{iterate} has $O(n^{3})$ computational complexity, while the computational complexity for the current low-rank algorithms is $O(nr^{2})$~\cite{Browet:2014, Marchand:2017}.  While for any $k$ in this ideal case the vectors, $v_{j,k}$, $j= 1, \dotsc, q$, associated with each role are not parallel, there is value to not simply taking $k$ some convenient small fixed value.
As $k$ increases the angles between the non-parallel vectors, $v_{j,k}$, $j= 1, \dotsc, q$, increase thereby increasing the discrimination capabilities when a low-rank approximation of the similarity matrix of a non-ideal adjacency matrix is used to select the nearby ideal adjacency matrix used to extract the role structure.

\begin{remark}\label{rem:QA}
If for a given adjacency matrix $A$ (not necessarily ideal) and a given assignment function $\sigma$, we use the cost function
$$ Q_A(\sigma,B):=\|A-A_{ideal}\|_F^2=\|A-(PZ)B(PZ)^T\|_F^2=\|P^TAP-ZBZ^T\|_F^2,
$$
then $P$ and $Z$ are completely defined by the assignment function $\sigma$. Therefore, once $\sigma$ is fixed, the above function can be decoupled as
the sum 
$$  Q_A(\sigma,B)=\sum_{i,j} \|(P^TAP)_{i,j} -B_{i,j}z_iz_j^T\|_F^2, 
$$
where $(P^TAP)_{i,j}$ is the $(i,j)$ block of the permuted matrix $A$. Clearly this is minimized by choosing $B_{i,j}=1$ if
$z_i^T(P^TAP)_{i,j}z_j > \frac{n_in_j}{2}$ and  $B_{i,j}=0$ otherwise. For a non-ideal adjacency matrix $A$ this leaves the key question of using an approximate low-rank factorization of its associated similarity matrix to determine $Z$ that, in part, defines a near-by ideal matrix for use in determining the role structure. The feasibility of such approximations yielding a useful $Z$ and therefore $B$ is considered in Section~\ref{sec:perturb}. \hfill $\Box$
\end{remark}

\section{Perturbation Analysis}\label{sec:perturb}
In this section, we analyze the singular values of the adjacency matrix $A$ of a directed unweighted ideal graph and the effect perturbing $A$ has on them and on the similarity matrix. The perturbed adjacency matrix is denoted as $A + \Delta$, where $\Delta$ is the perturbation (i.e., addition or subtraction) of some elements $a_{i,j}$. The main questions to be addressed are:
\begin{enumerate}
\item Can we estimate the number $q$ of ``ideal'' roles from the singular values of the perturbed graph ? 
\item Is the dominant subspace of the perturbed matrices $S_k$ then still close to $\Image Z$ of the ideal adjacency matrix $A$ 
so that we can find the correct grouping of nodes? 
\end{enumerate}
If we can answer these two questions affirmatively, then we can recover the correct grouping of nodes and their roles.

\subsection{The Singular Values of the Ideal Graph Case}
Let $A=(PZ)B(PZ)^T$ be the adjacency matrix of an  directed unweighted ideal graph where $B$ is a minimal role matrix. Then, $Z$ can be represented by the factorization $Z = U_q N^{1/2}$ where $U_q^{T} U_q = I_{q}$ and $N^{1/2}  := \diag\{\lVert z_{1} \rVert_{2}, \ldots, \lVert z_{q} \rVert_{2}\}$, and we can write $\tilde{S}_{k} := (PU_q)^T S_{k} (PU_q)$, where $\tilde{S}_{k}$ is symmetric and $\tilde{S}_{k} \succeq 0$. We can reformulate the iteration on the $S_k$ matrices as a recurrence for the $\tilde{S}_{k}$ matrices, i.e., 
\begin{equation} \label{BBT} 
\tilde{S}_1  := \left[\begin{array}{cc} \tilde{B} & \tilde{B}^T \end{array}\right] \left[\begin{array}{cc} \tilde{B}^T \\ \tilde{B} \end{array}\right],  \quad\tilde{S}_{k+1} := \left[\begin{array}{cc} \tilde{B} & \tilde{B}^T \end{array}\right]
\left[\begin{array}{cc} I+\beta^2\tilde{S}_k & 0 \\ 0 & I+\beta^2\tilde{S}_k   \end{array}\right]\left[\begin{array}{cc} \tilde{B}^T \\ \tilde{B} \end{array}\right] ,
\end{equation}
where $\tilde{B} := N^{1/2} B N^{1/2}$. Therefore $\Image \tilde{S}_k = \Image [\tilde{B}~\tilde{B}^T]$,  $\Image S_k= \Image PU_q\tilde{S}_k = \Image PU_q[\tilde{B}~\tilde{B}^T]$. Also, the nonzero singular values of $A$ and $S_{k}$ are those of
\begin{align}
\Sigma(A) &= \Sigma(\tilde{B}),~~\text{since } A = (PU_q)N^{1/2}BN^{1/2}(PU_q)^{T} \nonumber \\
\text{and} \; \;\Sigma(S_{k}) &= \Sigma(\tilde{S}_{k}),~~\text{since } S_{k} = (PU_q) \tilde{S}_{k} (PU_q)^T, \;  \forall k\ge 1 . \nonumber
\end{align}

\begin{remark}
It is informative to look at the case 
of ideal undirected graphs, since then $A=A^T$ and $S_1=2A^2$. Moreover, all matrices $S_k$ then commute with $A$ and the expression for the matrices $S_k$ simplifies to 
$$S_k=2A^2\sum_{\ell =1}^k [2\beta^2A^2]^{(\ell-1)},$$
provided $2\beta^2\rho(A^2)<1$. 
So if we denote the ordered singular values of $A$ and of $S_k^\frac12$ by
$\alpha_i$ and $\lambda_i^{(k)}$, respectively, then the following relations hold~:
\begin{equation} \label{lambdak}
 [\lambda_i^{(1)}]^2=2\alpha_i^2, \quad [\lambda_i^{(k)}]^2=[\lambda_i^{(1)}]^2\sum_{\ell =1}^k [\beta \lambda_i^{(1)}]^{2(\ell-1)},
\end{equation}
provided $\beta< 1/\lambda^{(1)}_{\max}$. It follows then that if $\lambda_i^{(1)}>\lambda_j^{(1)}$ then also $\lambda_i^{(k)}>\lambda_j^{(k)}$ for all $k$, as long as the condition on $\beta$ holds, since then $\sum_{\ell =1}^k [\beta \lambda_i^{(1)}]^{2(\ell-1)} > \sum_{\ell =1}^k [\beta \lambda_j^{(1)}]^{2(\ell-1)}$. In Appendix A, we also prove that 
$$   \lambda_i^{(1)}>\lambda_j^{(1)} \quad \Longrightarrow \quad    \frac{\lambda_i^{(k+1)}}{\lambda_j^{(k+1)}} >  \frac{\lambda_i^{(k)}}{\lambda_j^{(k)}},   \; \forall k \ge 1, 
$$
and that the ratio $\frac{\lambda_i^{(k)}}{\lambda_j^{(k)}}$ also grows with $\beta$ as long as it satisfies the bound  $\beta<1/\lambda^{(1)}_{\max}$.
Therefore the ordering of the singular values $\lambda_i^{(k)}$ is preserved for all $k$ and the relative gap between dominant singular values and the small ones, is easier to identify for large $k$ and for $\beta$ close to its upper bound $1/\lambda^{(1)}_{\max}$.  While these trends are specific to the undirected graph case, due to our empirical studies, we also use them for the directed graph case below. \hfill $\Box$
\end{remark}

\subsection{The Singular Values of $A$ for the Perturbed Graph Case}\label{subsect:PerturbOfA}

Due to the analysis above showing the relevant gap increasing  with $k$, we concentrate on comparing the singular values of $A$ with those of  $S_1$. 
Furthermore, if we write  $S_1 = XX^T$, where $X:= \left[\begin{array}{cc} A & A^T \end{array}\right]$, $S_1^\frac12$ has the same nonzero singular values as $X$.
Therefore, we analyze here the perturbations of the spectrum of $A$ and $X:= \left[\begin{array}{cc} A & A^T \end{array}\right]$ for an arbitrary perturbation $\Delta$ of the ideal    
adjacency matrix $A$ of a directed unweighted ideal graph.  We analyze the case of an arbitrary perturbation of $\{0,1\}$ type and the special case of a Erd\"{o}s-Renyi type perturbation. For simplicity, we assume that $A$ and $X$ have both the same rank $q$, which is the number of roles.

Since we can rewrite the ideal decomposition $A=ZBZ^T$ as the normalized factorization $A = U_q\tilde BU_q^T$, where $N=\diag(n_1,\ldots, n_q)$, $Z=U_qN^{\frac12}$, $U_q^TU_q=I_q$  and $\tilde B=N^{\frac12}BN^{\frac12}$, it follows that we can construct an orthogonal transformation 
\begin{equation*}
U = \left[\begin{array}{cc} U_q & U_q^\perp \end{array} \right]
\end{equation*}
such that 
\begin{equation*}
A = U \left[\begin{array}{cc} \tilde B & 0 \\ 0 & 0 \end{array} \right]U^T, \quad \mathrm{and} \quad    \Delta = U \left[\begin{array}{cc} \tilde \Delta_{11}  & \tilde \Delta_{12}  \\ \tilde \Delta_{21}  & \tilde \Delta_{22}  \end{array} \right]U^T. 
\end{equation*}
We are interested in the dominant singular values of the perturbed matrices $A(\Delta):=A+\Delta$ and $S_1^{\frac12}(\Delta):=[A(\Delta)A^T(\Delta)+A^T(\Delta)A(\Delta)]^{\frac12}$, or equivalently of $X(\Delta):=[A(\Delta), A^T(\Delta)]$ since 
 $S_1(\Delta)=X(\Delta)X^T(\Delta)$. Clearly the gap between the $q$-th and $(q+1)$-st  singular value will affect how well the number of roles $q$ is detected when the graph is perturbed. Since we assumed that $\tilde B$ has full rank $q$, the nonzero singular values of $A$ are those of $\tilde B$ and the nonzero singular values of $S_1^{\frac12}$ are those of $\tilde X:=[\;\tilde B \;\; \tilde B^T]$. We also assume that the perturbation $\Delta$ is sufficiently smaller than the norm of the ideal adjacency matrix $A$ such that one can use classical perturbation analysis techniques.

It follows then from standard perturbation theory of the matrix
\begin{equation*}
A(\Delta):= U \tilde A(\Delta) U^T \quad \mathrm{where} \quad \tilde A(\Delta):= \left[\begin{array}{cc} \tilde B+ \tilde \Delta_{11}  & \tilde \Delta_{12}  \\ \tilde \Delta_{21}  & \tilde \Delta_{22} \end{array} \right], 
\end{equation*}
that the  $q$  dominant singular values of $A(\Delta)$ are  $\| \Delta\|_2$-close to those of $\tilde B$ and that the 
$(q+1)$-st singular value of $A(\Delta)$ is strictly bounded by the norm of the submatrix $\left[\begin{array}{cc} \tilde \Delta_{21}  & \tilde \Delta_{22} \end{array} \right]$ (see e.g.~\cite{GVL:2013}). Similarly, the $q$ dominant singular values of 
\begin{equation*}
X(\Delta) =  U \left[\begin{array}{cc} \tilde A(\Delta)  & \tilde A^T(\Delta)\end{array} \right] (I_2\otimes U^T) ,
\end{equation*}
where
\begin{equation*}
\left[\begin{array}{cc} \tilde A(\Delta)  & \tilde A^T(\Delta)\end{array} \right]:= \left[\begin{array}{cccc} \tilde B+\tilde \Delta_{11}  & \tilde \Delta_{12} & \tilde B^T + \tilde \Delta_{11}^T  & \tilde \Delta_{21}^T  \\ \tilde \Delta_{21}  & \tilde \Delta_{22} & \tilde \Delta_{12}^T  & \tilde \Delta_{22}^T \end{array} \right], 
\end{equation*}
are then $\| \Delta\|_2$-close to those of $\left[\begin{array}{cc} \tilde B & \tilde B^T\end{array}\right]$ and the $(q+1)$-st singular value is strictly bounded by the norm of the submatrix 
$\left[\begin{array}{cccc} \tilde \Delta_{21}  & \tilde \Delta_{22} & \tilde \Delta_{12}^T  & \tilde \Delta_{22}^T \end{array} \right]$~\cite{GVL:2013}.

Note that the singular values of $\tilde B:=N^\frac12 B N^\frac12$ can be expected to be large since the diagonal scaling $N$ has large entries $n_i$,
the matrix $B$ has only $0$ and $1$ entries and is nonsingular. A precise lower bound on the smallest singular value of $B$ is not available but due to the $0, \;1$ structure it is expected to be $O(1)$ with an acceptable gap to the size of a reasonable perturbation. In practice, $\tilde B$ will have dimension small enough so its singular values can be computed with negligible additional cost to assess its quality in any role extraction algorithm.

Given an ideal adjacency matrix and a minimal $B$, we also consider an
Erd\"{o}s-R\'{e}nyi perturbation model determined by two probabilities $p_{in}$ and
$p_{out}$ \cite{EKYY:2013}.
Elements of $A$ that are $1$ change to $0$ with probability $p_{in}$ and
elements of $A$ that are $0$ change to $1$ with probability $p_{out}$.
Given this Erd\"{o}s-R\'{e}nyi perturbation model, the expected value of $\Delta$ is known~:
\begin{equation*}
E(\Delta) = U_P N^{\frac12} [B(1-p_{in}) + p_{out}(\mathbf{1}\mathbf{1}^T-B)] N^{\frac12}U_P^T
\end{equation*}
which implies that $(U_q^{\perp})^T E(\Delta)=0$ and  $E(\Delta){U_q^\perp} =0$. This then means that the norms
 of $\tilde \Delta_{12}$,  $\tilde \Delta_{21}$ and $\tilde \Delta_{22}$, can be expected to be much smaller that the norm of
 $\Delta$. 
 Moreover, this suggests that we can estimate the $q$ largest singular values of $A+\Delta$ by those of 
\begin{equation*}
E(\tilde B) = \tilde B +  N^{\frac12} [p_{in}B + p_{out}(\mathbf{1}\mathbf{1}^T-B)] N^{\frac12}
\end{equation*}
which are very close to those of $\tilde B$ when $p_{in}$ and $p_{out}$ are small.

For example, for the Erd\"{o}s-R\'{e}nyi graph in Figure~\ref{fig:BlockCycleGraph}, there is a distinct gap between the  $4$th and $5$th singular values for $A$, $S_{\infty}^{1/2}$ and $S_{\infty}$ for the ideal graph case (i.e., rows (a) and (c)). This indicates that the rank of $S_{\infty}$ is $4$, which is the number of roles. For the perturbed graph case, the gap between the  $4$th and $5$th singular values is smaller; however, the gap is larger in $S_{\infty}$ than  in $A$. Also, note that for the perturbed case, the difference between the $4$th and $5$th singular values of $S_{\infty}$ for the large graph is $10^{4}$, while it is $10^{2}$ for the smaller graph. So perturbing small graphs has a larger effect on how well the similarity measure can detect the number of roles than perturbing large graphs. However, in practice we are more interested in extracting structure for large graphs so this is not a major concern.  For both cases it is seen that using the similarity measure to detect the roles in the graph is preferable to using the adjacency matrix because the gap in the singular values is  larger for $S_{\infty}$ than it is for $A$.
Notice that this is intimately related to the fact that we consider unweighted adjacency matrices since 
those are the ones where the Erd\"{o}s-R\'{e}nyi property applies.

\begin{figure}[htbp!]\centering
\begin{subfigure}{0.25\textwidth}\centering
\includegraphics[scale=0.7]{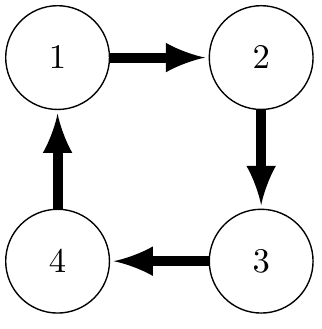}
 \end{subfigure}
\begin{subfigure}{0.25\textwidth}\centering
\begin{equation*}
B = \begin{bmatrix} 0 & 1 & 0 & 0 \\ 0 & 0 & 1 & 0 \\ 0 & 0 & 0 & 1\\ 1 & 0 & 0 & 0\\ \end{bmatrix}
\end{equation*}
 \end{subfigure}

\begin{subfigure}{1\textwidth}\centering
\begin{subfigure}{0.31\textwidth}\centering
\setlength\figurewidth{2.5cm}
\setlength\figureheight{2.5cm}
\pgfplotsset{
	xlabel style = {font = \tiny, yshift=1.15ex},
	ylabel style = {font =\normalsize,  yshift=-1.15ex},
	x tick label style = {font=\scriptsize\color{black}},
	y tick label style = {font=\scriptsize\color{black}},
	title style={font=\footnotesize, yshift=-1.5ex},
	ylabel = {(a)}
	}
%
\definecolor{mycolor1}{rgb}{0.00000,0.44700,0.74100}%
\begin{tikzpicture}

\begin{axis}[%
width=0.95092\figurewidth,
height=\figureheight,
at={(0\figurewidth,0\figureheight)},
scale only axis,
xmin=1,
xmax=10,
ymode=log,
ymin=1e-15,
ymax=1000000,
yminorticks=true,
title = {$\sigma(A)$},
]
\addplot [color=black,mark size=1.7pt,only marks,mark=*,mark options={solid},forget plot]
  table[row sep=crcr]{%
1	200\\
2	141.421356237309\\
3	141.421356237309\\
4	99.9999999999998\\
5	2.08192986419419e-12\\
6	1.77453123322964e-12\\
7	1.56022517173823e-12\\
8	1.44113022280848e-12\\
9	1.36186919200412e-12\\
10	1.15925439729614e-12\\
};
\end{axis}
\end{tikzpicture}%
 \end{subfigure}
\begin{subfigure}{0.3\textwidth}\centering
\setlength\figurewidth{2.5cm}
\setlength\figureheight{2.5cm}
\pgfplotsset{
	xlabel style = {font = \tiny, yshift=1.15ex},
	ylabel style = {font =\tiny},
	x tick label style = {font=\scriptsize\color{black}},
	y tick label style = {font=\scriptsize\color{black}},
	title style={font=\footnotesize, yshift=-1.5ex},
	}
%
\definecolor{mycolor1}{rgb}{0.00000,0.44700,0.74100}%
\begin{tikzpicture}

\begin{axis}[%
width=0.968862\figurewidth,
height=\figureheight,
at={(0\figurewidth,0\figureheight)},
scale only axis,
xmin=1,
xmax=10,
ymode=log,
ymin=1e-15,
ymax=1000000,
yminorticks=true,
title = {$\sigma(S_{\infty}^{1/2})$},
]
\addplot [color=black,mark size=1.7pt,only marks,mark=*,mark options={solid},forget plot]
  table[row sep=crcr]{%
1	382.437241958163\\
2	382.437241958162\\
3	270.423967166897\\
4	270.423967166897\\
5	2.23906952243018e-05\\
6	1.44873762511274e-05\\
7	1.19446927100933e-05\\
8	1.04305309335493e-05\\
9	8.66664716042303e-06\\
10	8.5665826800181e-06\\
};
\end{axis}
\end{tikzpicture}%
\end{subfigure}
~ 
\begin{subfigure}{0.3\textwidth}\centering
\setlength\figurewidth{2.5cm}
\setlength\figureheight{2.5cm}
\pgfplotsset{
	xlabel style = {font = \tiny, yshift=1.15ex},
	ylabel style = {font =\tiny},
	x tick label style = {font=\scriptsize\color{black}},
	y tick label style = {font=\scriptsize\color{black}},
	title style={font=\footnotesize, yshift=-1.5ex},
	}
%
\definecolor{mycolor1}{rgb}{0.00000,0.44700,0.74100}%
\begin{tikzpicture}

\begin{axis}[%
width=0.95092\figurewidth,
height=\figureheight,
at={(0\figurewidth,0\figureheight)},
scale only axis,
xmin=1,
xmax=10,
ymode=log,
ymin=1e-15,
ymax=1000000,
yminorticks=true,
title = {$\sigma(S_{\infty})$},
]
\addplot [color=black,mark size=1.7pt,only marks,mark=*,mark options={solid},forget plot]
  table[row sep=crcr]{%
1	146258.244036566\\
2	146258.244036566\\
3	73129.1220182832\\
4	73129.1220182828\\
5	1.47232257408719e-09\\
6	1.25832148741443e-09\\
7	1.23204556798165e-09\\
8	1.03897886686468e-09\\
9	1.01265598020761e-09\\
10	9.92011217618649e-10\\
};
\end{axis}
\end{tikzpicture}%
\end{subfigure}

\vfill

\begin{subfigure}{0.3\textwidth}\centering
\setlength\figurewidth{2.5cm}
\setlength\figureheight{2.5cm}
\pgfplotsset{
	xlabel style = {font = \tiny, yshift=1.15ex},
	ylabel style = {font =\normalsize,  yshift=-1.15ex},
	x tick label style = {font=\scriptsize\color{black}},
	y tick label style = {font=\scriptsize\color{black}},
	title style={font=\footnotesize, yshift=-1.5ex},
	ylabel = {(b)},
	}
%
\definecolor{mycolor1}{rgb}{0.00000,0.44700,0.74100}%
\begin{tikzpicture}

\begin{axis}[%
width=0.982769\figurewidth,
height=\figureheight,
at={(0\figurewidth,0\figureheight)},
scale only axis,
xmin=1,
xmax=10,
ymode=log,
ymin=0,
ymax=400000,
yminorticks=true,
title = {$\sigma(A)$},
]
\addplot [color=black,mark size=1.7pt,only marks,mark=*,mark options={solid},forget plot]
  table[row sep=crcr]{%
1	215.406742305397\\
2	104.699041707322\\
3	85.1686307254007\\
4	66.3537149730647\\
5	19.4509380239593\\
6	19.3443953236945\\
7	19.0677161156926\\
8	18.9254200934063\\
9	18.8343350848519\\
10	18.7469910867677\\
};
\end{axis}
\end{tikzpicture}%
 \end{subfigure}
\begin{subfigure}{0.3\textwidth}\centering
\setlength\figurewidth{2.45cm}
\setlength\figureheight{2.45cm}
\pgfplotsset{
	xlabel style = {font = \tiny, yshift=1.15ex},
	ylabel style = {font =\tiny},
	x tick label style = {font=\scriptsize\color{black}},
	y tick label style = {font=\scriptsize\color{black}},
	title style={font=\footnotesize, yshift=-1.5ex},
	}
%
\definecolor{mycolor1}{rgb}{0.00000,0.44700,0.74100}%
\begin{tikzpicture}

\begin{axis}[%
width=0.982769\figurewidth,
height=\figureheight,
at={(0\figurewidth,0\figureheight)},
scale only axis,
xmin=1,
xmax=10,
ymode=log,
ymin=0,
ymax=400000,
yminorticks=true,
title = {$\sigma(S_{\infty}^{1/2})$},
]
\addplot [color=black,mark size=1.7pt,only marks,mark=*,mark options={solid},forget plot]
  table[row sep=crcr]{%
1	554.327201940954\\
2	169.674421624819\\
3	132.661778894335\\
4	119.57142313491\\
5	24.089203877101\\
6	23.7232112725706\\
7	23.4286173976938\\
8	23.1809870369831\\
9	23.0701457424965\\
10	23.0143164209567\\
};
\end{axis}
\end{tikzpicture}%
\\
\end{subfigure}
\begin{subfigure}{0.3\textwidth}\centering
\setlength\figurewidth{2.5cm}
\setlength\figureheight{2.5cm}
\pgfplotsset{
	xlabel style = {font = \tiny, yshift=1.15ex},
	ylabel style = {font =\tiny},
	x tick label style = {font=\scriptsize\color{black}},
	y tick label style = {font=\scriptsize\color{black}},
	title style={font=\footnotesize, yshift=-1.5ex},
	}
%
\definecolor{mycolor1}{rgb}{0.00000,0.44700,0.74100}%
\begin{tikzpicture}

\begin{axis}[%
width=0.982769\figurewidth,
height=\figureheight,
at={(0\figurewidth,0\figureheight)},
scale only axis,
xmin=1,
xmax=10,
ymode=log,
ymin=0,
ymax=400000,
yminorticks=true,
title = {$\sigma(S_{\infty})$},
]
\addplot [color=black,mark size=1.7pt,only marks,mark=*,mark options={solid},forget plot]
  table[row sep=crcr]{%
1	307278.646811686\\
2	28789.4093537166\\
3	17599.1475794093\\
4	14297.3252305077\\
5	580.289743432522\\
6	562.790753082995\\
7	548.900113167504\\
8	537.358160008774\\
9	532.231624580028\\
10	529.658760323919\\
};
\end{axis}
\end{tikzpicture}%
\\
\end{subfigure}
\caption{200, 100, 100, and 200 nodes in each role}
\end{subfigure}

\vfill

\begin{subfigure}{1\textwidth}\centering
\begin{subfigure}{0.31\textwidth}\centering
\setlength\figurewidth{2.5cm}
\setlength\figureheight{2.5cm}
\pgfplotsset{
	xlabel style = {font = \tiny, yshift=1.15ex},
	ylabel style = {font =\normalsize,  yshift=-1.15ex},
	x tick label style = {font=\scriptsize\color{black}},
	y tick label style = {font=\scriptsize\color{black}},
	title style={font=\footnotesize, yshift=-1.5ex},
	ylabel = {(c)}
	}
%
\definecolor{mycolor1}{rgb}{0.00000,0.44700,0.74100}%
\begin{tikzpicture}

\begin{axis}[%
width=0.95092\figurewidth,
height=\figureheight,
at={(0\figurewidth,0\figureheight)},
scale only axis,
xmin=1,
xmax=10,
ymode=log,
ymin=1e-15,
ymax=2000,
yminorticks=true,
title = {$\sigma(A)$},
]
\addplot [color=black,mark size=1.7pt,only marks,mark=*,mark options={solid},forget plot]
  table[row sep=crcr]{%
1	20\\
2	14.142135623731\\
3	14.142135623731\\
4	10\\
5	1.39588404348388e-14\\
6	1.26885171848291e-14\\
7	1.12538161041208e-14\\
8	9.20043615849947e-15\\
9	8.75443058862176e-15\\
10	8.17660993421369e-15\\
};
\end{axis}
\end{tikzpicture}%
 \end{subfigure}
\begin{subfigure}{0.3\textwidth}\centering
\setlength\figurewidth{2.5cm}
\setlength\figureheight{2.5cm}
\pgfplotsset{
	xlabel style = {font = \tiny, yshift=1.15ex},
	ylabel style = {font =\tiny},
	x tick label style = {font=\scriptsize\color{black}},
	y tick label style = {font=\scriptsize\color{black}},
	title style={font=\footnotesize, yshift=-1.5ex},
	}
%
\definecolor{mycolor1}{rgb}{0.00000,0.44700,0.74100}%
\begin{tikzpicture}

\begin{axis}[%
width=0.968862\figurewidth,
height=\figureheight,
at={(0\figurewidth,0\figureheight)},
scale only axis,
xmin=1,
xmax=10,
ymode=log,
ymin=1e-15,
ymax=2000,
yminorticks=true,
title = {$\sigma(S_{\infty}^{1/2})$},
]
\addplot [color=black,mark size=1.7pt,only marks,mark=*,mark options={solid},forget plot]
  table[row sep=crcr]{%
1	38.2437241883514\\
2	38.2437241883514\\
3	27.0423967114113\\
4	27.0423967114113\\
5	5.32435699775986e-07\\
6	5.32435699775986e-07\\
7	3.90523048141949e-07\\
8	3.90523048052908e-07\\
9	3.64838139519924e-07\\
10	3.64838139519924e-07\\
};
\end{axis}
\end{tikzpicture}%
\end{subfigure}
\begin{subfigure}{0.3\textwidth}\centering
\setlength\figurewidth{2.5cm}
\setlength\figureheight{2.5cm}
\pgfplotsset{
	xlabel style = {font = \tiny, yshift=1.15ex},
	ylabel style = {font =\tiny},
	x tick label style = {font=\scriptsize\color{black}},
	y tick label style = {font=\scriptsize\color{black}},
	title style={font=\footnotesize, yshift=-1.5ex},
	}
%
\definecolor{mycolor1}{rgb}{0.00000,0.44700,0.74100}%
\begin{tikzpicture}

\begin{axis}[%
width=0.95092\figurewidth,
height=\figureheight,
at={(0\figurewidth,0\figureheight)},
scale only axis,
xmin=1,
xmax=10,
ymode=log,
ymin=1e-15,
ymax=2000,
yminorticks=true,
title = {$\sigma(S_{\infty})$},
]
\addplot [color=black,mark size=1.7pt,only marks,mark=*,mark options={solid},forget plot]
  table[row sep=crcr]{%
1	1462.5824397947\\
2	1462.5824397947\\
3	731.291219897349\\
4	731.291219897349\\
5	1.05362068823597e-12\\
6	9.82968470157355e-13\\
7	8.17836950257446e-13\\
8	6.14125752965437e-13\\
9	6.05529173119955e-13\\
10	5.23301131306281e-13\\
};
\end{axis}
\end{tikzpicture}%
\end{subfigure}

\vfill

\begin{subfigure}{0.3\textwidth}\centering
\setlength\figurewidth{2.5cm}
\setlength\figureheight{2.5cm}
\pgfplotsset{
	xlabel style = {font = \tiny, yshift=1.15ex},
	ylabel style = {font =\normalsize,  yshift=-1.15ex},
	x tick label style = {font=\scriptsize\color{black}},
	y tick label style = {font=\scriptsize\color{black}},
	title style={font=\footnotesize, yshift=-1.5ex},
	ylabel = {(d)},
	}
%
\definecolor{mycolor1}{rgb}{0.00000,0.44700,0.74100}%
\begin{tikzpicture}

\begin{axis}[%
width=0.982769\figurewidth,
height=\figureheight,
at={(0\figurewidth,0\figureheight)},
scale only axis,
xmin=1,
xmax=10,
ymode=log,
ymin=0,
ymax=4000,
yminorticks=true,
title = {$\sigma(A)$},
]
\addplot [color=black,mark size=1.7pt,only marks,mark=*,mark options={solid},forget plot]
  table[row sep=crcr]{%
1	22.3380940883889\\
2	11.3184227175985\\
3	9.28726896330278\\
4	7.87532484990907\\
5	5.83703140752192\\
6	5.6830193088767\\
7	5.41418590938828\\
8	5.19944057428238\\
9	5.1463225867579\\
10	4.83570734794972\\
};
\end{axis}
\end{tikzpicture}%
 \end{subfigure}
\begin{subfigure}{0.3\textwidth}\centering
\setlength\figurewidth{2.5cm}
\setlength\figureheight{2.5cm}
\pgfplotsset{
	xlabel style = {font = \tiny, yshift=1.15ex},
	ylabel style = {font =\tiny},
	x tick label style = {font=\scriptsize\color{black}},
	y tick label style = {font=\scriptsize\color{black}},
	title style={font=\footnotesize, yshift=-1.5ex},
	}
%
\definecolor{mycolor1}{rgb}{0.00000,0.44700,0.74100}%
\begin{tikzpicture}

\begin{axis}[%
width=0.982769\figurewidth,
height=\figureheight,
at={(0\figurewidth,0\figureheight)},
scale only axis,
xmin=1,
xmax=10,
ymode=log,
ymin=0,
ymax=4000,
yminorticks=true,
title = {$\sigma(S_{\infty}^{1/2})$},
]
\addplot [color=black,mark size=1.7pt,only marks,mark=*,mark options={solid},forget plot]
  table[row sep=crcr]{%
1	56.6133703533563\\
2	17.9441789665624\\
3	14.34794814331\\
4	13.4020708820768\\
5	7.78822191207666\\
6	7.27874613775609\\
7	6.94248507510714\\
8	6.65909075135713\\
9	6.6315185566224\\
10	6.44507666456352\\
};
\end{axis}
\end{tikzpicture}%
\\
\end{subfigure}
\begin{subfigure}{0.3\textwidth}\centering
\setlength\figurewidth{2.5cm}
\setlength\figureheight{2.5cm}
\pgfplotsset{
	xlabel style = {font = \tiny, yshift=1.15ex},
	ylabel style = {font =\tiny},
	x tick label style = {font=\scriptsize\color{black}},
	y tick label style = {font=\scriptsize\color{black}},
	title style={font=\footnotesize, yshift=-1.5ex},
	}
%
\definecolor{mycolor1}{rgb}{0.00000,0.44700,0.74100}%
\begin{tikzpicture}

\begin{axis}[%
width=0.982769\figurewidth,
height=\figureheight,
at={(0\figurewidth,0\figureheight)},
scale only axis,
xmin=1,
xmax=10,
ymode=log,
ymin=0,
ymax=4000,
yminorticks=true,
title = {$\sigma(S_{\infty})$},
]
\addplot [color=black,mark size=1.7pt,only marks,mark=*,mark options={solid},forget plot]
  table[row sep=crcr]{%
1	3205.07370276626\\
2	321.993558784018\\
3	205.863615923114\\
4	179.615503928212\\
5	60.6564005517509\\
6	52.9801453378983\\
7	48.1980990180845\\
8	44.34348963481\\
9	43.977038366827\\
10	41.5390132121015\\
};
\end{axis}
\end{tikzpicture}%
\\
\end{subfigure}
\caption{20, 10, 10, and 20 nodes in each role}
\end{subfigure}
\caption{Block cycle role structure, associated neighborhood pattern similarity, and $10$ largest singular values of $A$, $S_{\infty}^{1/2}$, and $S_{\infty}$. Rows (a) and (c) are the singular values for an ideal graph. Rows (b) and (d) are the singular values for a perturbed graph.  \label{fig:BlockCycleGraph}}
\end{figure}

\section{Unification of Special Complex Structures as Role Structures}\label{sec:complex}

In this section, we unify special complex network structures as role structures and show why the neighborhood pattern similarity measure will extract these structures. We do this by examining the structure of the associated ideal adjacency matrix, $A$. As in the general discussions of earlier sections,  in efficient role extraction algorithms, these ideal forms would be identified by considering low-rank approximations of the similarity matrices of  non-ideal adjacency matrices.

\subsection{Community Structures}\label{subsect:Community}
A popular type of network structure is a community structure. Community structures are described as groups of nodes where there exists many connections between nodes in the same group and no (or few) connections between nodes in different groups \cite{Newman:2010}.  For the role extraction problem, community structures can be viewed as role structures, where the role matrix $B$ is an identity matrix of dimension equal to the number of roles.  Therefore, community structures are a special case of role structures. Since we assume, in the ideal case, that the intrarole adjacency matrices are cliques, the graph with ideal community structure is unweighted and undirected with a symmetric adjacency matrix $A$.
For example, a three community network can be represented by $A = (PZ) B (PZ)^T,$ where $A$ is an ideal graph with 
\begin{equation*}
B := \begin{bmatrix} 1 & 0 & 0 \\ 0 & 1 & 0 \\ 0 & 0 & 1 \\\end{bmatrix}.
\end{equation*}
Note that matrix $B$ is a minimal role matrix of the adjacency matrix $A$ and that $\rank (B) = 3$, which is the number of roles. Therefore, $\rank (S_{\infty}) = 3$ and we can recover the three communities.

\subsection{Overlapping Community Structures}\label{subsect:OverlapCommuntiy}
Another type of network structure is an overlapping community structure. Community structures emphasize the presence of dependencies inside a group and the absence of dependencies between groups.  However, there may exist nodes that can be placed in multiple communities without significantly altering the value of the cost function being minimized. That is, given two separate communities {\bf A} and {\bf B}, a third group of nodes {\bf C} may be included in either {\bf A} or {\bf B} if the cost function fails to determine a significance of one community over the other~\cite{PDFV:2005,Reichardt:2009}. Therefore, it can be concluded that {\bf A} and {\bf B} are overlapping communities and {\bf C} is the overlap (see Figure~\ref{fig:Overlap2RoleExampleA}).

\begin{figure}[ht!]\centering
\begin{subfigure}{0.33\textwidth}\centering
\includegraphics[scale = 0.7]{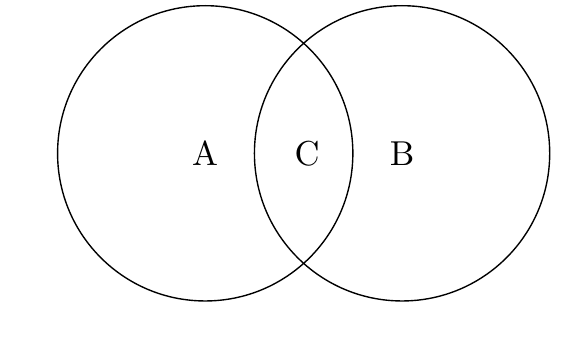}
\caption{Overlapping Community Structure \label{fig:Overlap2RoleExampleA}}
\end{subfigure}
~
\begin{subfigure}{0.3\textwidth}\centering
\includegraphics[scale = 0.14]{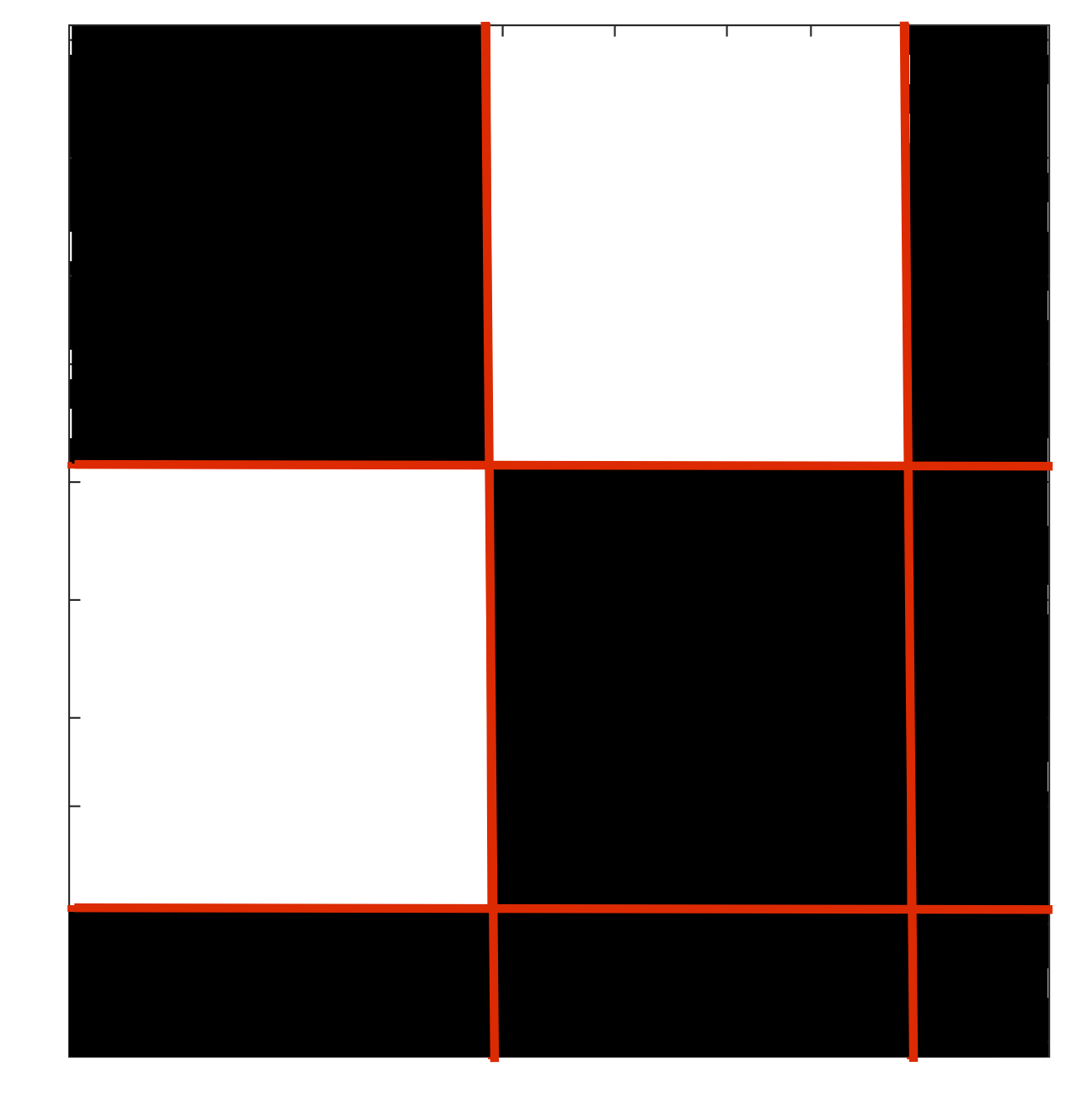}
\caption{Adjacency Matrix \label{fig:Overlap2RoleExampleB}}
\end{subfigure}
~
\begin{subfigure}{0.3\textwidth}\centering
\includegraphics[scale = 0.65]{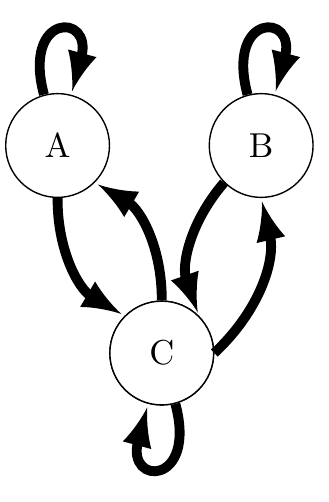}
\caption{Role Structure \label{fig:Overlap2RoleExampleC}}
\end{subfigure}
\caption{Example of two overlapping communities. \label{fig:Overlap2RoleExample}}
\end{figure}

Observe for the adjacency matrix (Figure~\ref{fig:Overlap2RoleExampleB}) that the overlapping community structure in Figure~\ref{fig:Overlap2RoleExampleA} can be represented by the $3$ role structure in Figure~\ref{fig:Overlap2RoleExampleC}. That is, the overlap {\bf C} can be represented by its own role where its role has connections to other nodes within the same role and to nodes in roles {\bf A} and {\bf B}. Also, the nodes in roles {\bf A} and {\bf B} do not  interact with each other and only interact with nodes within the same role or with nodes in role {\bf C}. Lastly, the above role structure is a valid role structure since it satisfies the role constraint that no two roles are structurally equivalent. 

For Figure~\ref{fig:Overlap2RoleExampleC}, assuming $A = (PZ) B (PZ)^T$  is an ideal graph,  the minimal role matrix $B$ is
\begin{equation*}
B := \begin{bmatrix} 1 & 0 & 1 \\ 0 & 1 & 1 \\ 1 & 1 & 1 \\ \end{bmatrix}
\end{equation*}
and $\rank(B) = 3$. Then, $\rank (S_{\infty}) = 3$ and we can recover the overlapping community structure.

In general, a role structure is considered an overlapping community structure when two (or more) roles interact with another role (the overlap role) and with themselves. In addition, the overlap role interacts with the other two (or more) roles and with itself. Also, the matrix $B$ will be a full rank minimal role matrix.

\subsection{Bipartite Networks and Communities}\label{subsect:Bipartite}
A common network in many applications, such as biological networks, is a bipartite network. A bipartite network is a set of nodes decomposed into two disjoint sets (say, of dimensions $n_1$ and $n_2$) such that no two nodes within the same set are adjacent \cite{Newman:2010}. For an appropriate ordering of the nodes, the adjacency matrix  $A$ can then be partitioned as
$$  A=\left[\begin{array}{cc}  0_{n_1} & A_{12} \\ A_{21} & 0_{n_2}  \end{array}\right].
$$
One easily checks that the similarity matrices $S_k$ are then block diagonal and of the form
$$
S_k= \left[\begin{array}{cc}  S_{11} & 0_{n_1,n_2} \\ 0_{n_2,n1} & S_{22} \end{array}\right].
$$
We are particularly interested in ideal bipartite networks where the communities are bipartite cliques (a clique with the
edges within each part removed).
Figure~\ref{fig:ExampleBipartiteRMP} is an example of a bipartite clique and its corresponding adjacency matrix. Observe that this bipartite clique has a very simple role structure (see Figure~\ref{fig:ExampleBipartiteRole}). This bipartite clique
has an ideal adjacency matrix, $A = (PZ)B(PZ)^T$ , with the minimal role matrix $B$ 
\begin{equation*}
B = \begin{bmatrix} 0 & 1 \\ 1 & 0 \\ \end{bmatrix},
\end{equation*}
and $\rank(B) = 2$. It follows that the similarity matrices, $S_k$ with $k \ge 1$, for this bipartite clique have $\rank(S_k) = 2$ and we recover the two roles from the role matrix.

 \begin{figure}[htb!]\centering
 \begin{subfigure}{0.33\textwidth}\centering
 \includegraphics[scale=0.5]{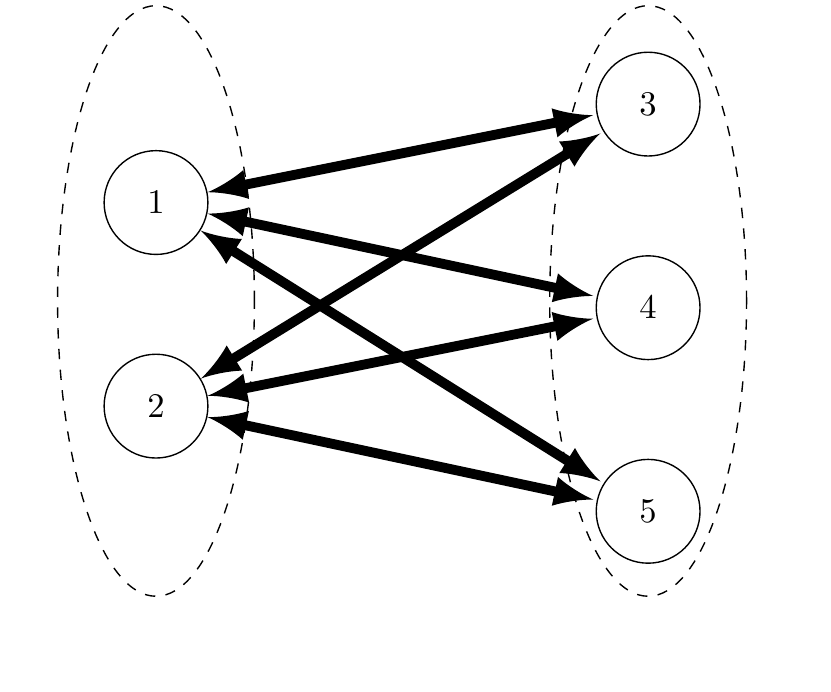}
  \caption{Bipartite clique \label{fig:ExampleBipartite}
} 
  \end{subfigure}
  ~ 
  \begin{subfigure}{0.30\textwidth}\centering
\begin{equation*}
A = \left[\begin{array}{cc|ccc}
0 & 0  & 1 & 1 & 1 \\ 
0 & 0  & 1 & 1 & 1 \\ \hline
1 & 1 & 0 & 0 & 0 \\
1 & 1 & 0 & 0 & 0 \\
1 & 1 & 0 & 0 & 0 \\ \end{array}\right]
\end{equation*}
\caption{Adjacency matrix of the bipartite clique\label{fig:ExampleBipartiteImageMatrix}} 
  \end{subfigure}
  ~ 
  \begin{subfigure}{0.30\textwidth}\centering
 \begin{equation*}
B = \left[\begin{array}{c|c}
0 & 1 \\  \hline
1 & 0  \end{array}\right]
\end{equation*}
  \caption{Role Structure} \label{fig:ExampleBipartiteRole}
  \end{subfigure}
 \caption{Example of a bipartite clique.  \label{fig:ExampleBipartiteRMP}}
\end{figure}

A set of bipartite communities is a bipartite network that can be reordered as a collection of bipartite cliques, each considered as one bipartite community. For example, the role matrix 
\begin{equation*}
B = \left[\begin{array}{ccc|ccc}
0 & 0 & 0 & 1 & 0 & 0 \\
0 & 0 & 0 & 0 & 1 & 0 \\
0 & 0 & 0 & 0 & 0 & 1 \\ \hline
1 & 0 & 0 & 0 & 0 & 0 \\
0 & 1 & 0 & 0 & 0 & 0 \\
0 & 0 & 1 & 0 & 0 & 0 \\ \end{array}\right]
\end{equation*}
has three bipartite communities. Also, observe that $B$ has $6$ roles  and is a minimal image matrix with $\rank(B)= 6$. Therefore, $\rank(S_k) = 6$ and we can recover the $6$ role structure, i.e., $3$ bipartite communities.

In general, a network with $q$ bipartite communities has a $2q\times 2q$ minimal role matrix $B$ with $0$  in the two $q \times q$ diagonal blocks and a $q\times q$ identity matrix in the two off-diagonal blocks.
The $q$ bipartite community network has $2q$ roles and since $B$ is symmetric,  $\rank(B) = \rank(S_k) = 2q$, which is equal to the number of roles. Therefore, the similarity matrix can be used to determine the role structure of bipartite communities. 
Since each community is based on a bipartite clique, as with communities in Section~\ref{subsect:Community},
this implies that the ideal adjacency matrix is symmetric.


\subsection{Signed Networks}
A signed unweighted directed graph is denoted  \\
$G(V, E^{-}, E^{+})$,
where $E^{-}\subseteq V\times V$ are the negative edges, $E^{+} \subseteq V\times V$ are the positive edges, and no edge can be both positive and negative (i.e., $E^{-} \cap E^{+} = \emptyset$) \cite{DBF:2005}. The associated signed unweighted adjacency matrix is defined in terms of its elements $A_{i, j}$ by
\begin{equation*}
A_{i, j} = \begin{cases} -1, & \text{ if } (i, j) \in E^{-}, \\ 
1, & \text{ if } (i, j) \in E^{+}, \\ 
0, & \text{ otherwise.} \end{cases}
\end{equation*}

Networks with positive and negative edge weights include social analysis networks and recommender networks~\cite{WF:1994, DBF:2005,Newman:2010,Traag:2014}. In such networks negative edges often denote a dislike towards a person, place or thing. 
Of interest here are the set of signed unweighted graphs that are ``checkerboard".
\begin{definition}\label{def:checkbrd}
A signed unweighted directed graph $G$ with signed adjacency matrix $A$ has a checkerboard pattern if there exists a diagonal sign matrix $Q \in \{0, \pm 1\}^{n \times n}$ satisfying $Q^2=I$ such that $ | A | = QAQ$.
\end{definition}
This is equivalent to there being a permutation matrix $P$ such that $PAP^T$ can be partitioned into a $2 \times 2$ block
structure where the two diagonal blocks contain only $0$ and $1$ elements and the two off-diagonal blocks contain
only $0$ and $-1$ elements.

A simple (partitioned) example of a checkerboard adjacency matrix, is
\begin{equation} \label{signed} 
A= \left[\begin{array}{r|rr|r|rr}
0 & 0 & -1 & 0  & 0 & 0 \\ \hline
0 & 0 & 1  & 0  & 0 & 0 \\ 
0 & 0 & 0  & -1 & 1 & 1 \\ \hline
1 & -1 & 0 & 0  & 0 & 0 \\ \hline
-1 & 1 & 0 & 0  & 0 & 0 \\
-1 & 1 & 0 & 0  & 0 & 0 
 \end{array}\right] , \;  QAQ=|A|= \left[\begin{array}{rrrrrr}
0 & 0 & 1 & 0  & 0 & 0 \\ 
0 & 0 & 1 & 0  & 0 & 0 \\
0 & 0 & 0 & 1  & 1 & 1 \\
1 & 1 & 0 & 0  & 0 & 0 \\
1 & 1 & 0 & 0  & 0 & 0 \\
1 & 1 & 0 & 0  & 0 & 0 
 \end{array}\right] ,
\end{equation}
where $Q :=\diag(1, -I_2, 1, -I_2)$ also indicates the checkerboard partitioning. 

Checkerboard matrices are related to so-called socially balanced networks, which were introduced by Heider in~\cite{Heider:1946,Heider:1958} and later analyzed by Cartwright and Haray in~\cite{CH:1956}. Further discussion of balanced networks and checkerboard graphs can be found in~\cite{WF:1994, DBF:2005, Traag:2014, Marchand:2017}. 

There is a simple relationship between the sequence of matrices $S_k$ defining the similarity matrix for a checkerboard signed adjacency matrix $A$ and the sequence of matrices $\tilde{S}_k$ defining
the similarity matrix of adjacency matrix $\tilde{A}=|A|$.  Specifically, if 
$QAQ=|A|=\tilde{A}$ for a diagonal sign matrix $Q$, satisfying $Q^2=I$, it then follows
from \eqref{iterate} that $\tilde{S}_k=QS_kQ=|S_k|$. This implies that the ranks and singular values of $S_k$ and $|S_k|$ are the same and that for the low-rank approximation one can as well consider the iteration matrices $|S_k|$ for the unsigned adjacency matrix $|A|$.

For checkerboard signed ideal adjacency matrices with the rank factorization
$A=ZBZ^T$ where $B\in \{0, 1\}^{q\times q}$ is a minimal role matrix 
and $Z\in \{0, 1, -1\}^{q\times q}$, i.e., the signs are placed in elements of $Z$,
the unsigned graph $|A|=QAQ$ also has (unsigned) ideal structure. This is easily seen
by considering the relevant factorizations of $A$ and $|A|$.
Suppose $A = ZBZ^T$, where $B \in \{0, \;1\}^{q \times q}$ is a minimal role matrix,
and $Z \in \{0, \;1,\;-1\}^{n \times q}$ where each row has exactly one nonzero element.
A diagonal sign matrix $Q$ satisfying $Q^2=I$ is easily constructed using the signs
of the single nonzero in each row of $Z$ so that $|A|=QAQ$ and $\tilde{Z} = QZ$ has a single
nonzero equal to $1$ in each row.  The unsigned ideal factorization of $|A|$ is therefore
$|A|=\tilde{Z} B \tilde{Z}^T$.

Assuming $A$ is checkerboard, i.e., a $Q$ is known, the other direction is also easily deduced.
If $|A|=QAQ=\tilde{Z} B \tilde{Z}^T$ is an unsigned ideal factorization with minimal $B$ then 
$
A=Q|A|Q=Q\tilde{Z} B \tilde{Z}^TQ^T= Z B Z^T
$
where $Z$ has a single
nonzero equal to $1$ or $-1$ in each row. So a checkerboard $A$ has  signed ideal form if and only if
$|A|$ has unsigned ideal form.

For \eqref{signed} we have 
$A= Z B Z^T$, $|A|= \tilde{Z} B \tilde{Z}^T$,
$\rank(A) = \rank(B)= q = 3$ and
$$Z = 
\left[\begin{array}{rrr}
1 & 0 & 0 \\
-1 & 0 & 0 \\
0 & -1 & 0 \\
0 & 0 & 1 \\
0 & 0 & -1 \\
0 & 0 & -1 
\end{array}\right],\quad 
B=
\begin{bmatrix}
0 & 1 & 0 \\ 0 & 0 & 1 \\ 1 & 0 & 0
\end{bmatrix}
\quad \tilde{Z}=QZ = 
\begin{bmatrix}
1 & 0 & 0 \\
1 & 0 & 0 \\
0 & 1 & 0 \\
0 & 0 & 1 \\
0 & 0 & 1 \\
0 & 0 & 1 
\end{bmatrix}.
$$

Since the vectors $z_i$ that define $Z$ and the node-to-role mapping may contain both $+1$ and $-1$, $B$ does not reflect the mixed sign checkerboard of $A$. If each role for which $z_i$ has both signs is split into two roles then a mixed sign checkerboard generalized role matrix, $\hat{B}$, is defined with dimension no larger than $2q \times 2q$. For this $A$, $\hat{B}=\hat{Z}B \hat{Z}^T$ has $5$ signed roles with
$$\hat{Z} = 
\left[\begin{array}{rrr}
1 & 0 & 0 \\
-1 & 0 & 0 \\
0 & -1 & 0 \\
0 & 0 & 1 \\
0 & 0 & -1 
\end{array}\right],\quad 
B=
\begin{bmatrix}
0 & 1 & 0 \\ 0 & 0 & 1 \\ 1 & 0 & 0
\end{bmatrix},
\quad
\hat{B}=
\left[\begin{array}{rrrrr}
0 & 0 & -1 & 0 & 0 \\
0 & 0 &  1 & 0 & 0 \\
0 & 0 &  0 & -1& 1 \\
1 & -1 &  0 & 0& 0 \\
-1 & 1 &  0 & 0& 0 
\end{array}\right].
$$

\subsection{Weighted Graphs}
In practice, many networks have edge weights in their graphs. Much of the theory developed above can be applied to a weighted matrix $W$ that is symmetric and rank one, i.e. $W=d d^T$. This is an example of a weighted adjacency matrix $A_W=W\circ A$ (where $\circ$ denotes the elementwise matrix product) that can be rewritten as $A_W:= DAD$, where $D=\diag(d_1, \ldots , d_n)$ and $A$ is the unweighted adjacency matrix.
For such weighted graphs, the adjacency matrix of the ideal graph case becomes 
$A_W=DAD=(DPZ)B(DPZ)^T$, with $B\in \left\{0,1\right\}^{q\times q}$.
If we use the permuted weight matrix $D_P=P^TDP$ and the corresponding scaled matrix $Z_D:=D_PZ$, we obtain a decomposition of the same type as for the unweighted case and with the same matrix $B$~:
\begin{equation*}
A_W=DAD=(DPZ)B(DPZ)^T= (PZ_D)B(PZ_D)^T = P(Z_DBZ_D^T)P^T. 
\end{equation*}
This shows that we should also be able to associate similarity matrices $S_k^D$ to a weighted matrix $A_W$.  
Since the effectiveness of the similarity matrices $S_k$ depends on the connection between adjacency matrices and Erd\"{o}s-R\'{e}nyi graphs, we want to maintain this connection in the scaled similarity matrices $S^D_k$, i.e. $S^D_k=DS_kD$. One then finds that the corresponding 
recurrences for the matrices  $S^D_k$ are given by
$$
S_1^D:= \left[\begin{array}{cc} A_W & A_W^T \end{array}\right]\left[\begin{array}{cc} D^{-2} & 0 \\ 0 & D^{-2}\end{array}\right] \left[\begin{array}{cc} A_W^T \\ A_W \end{array}\right], $$ 
$$S_{k+1}^D = \left[\begin{array}{cc} A_W & A^T_W \end{array}\right]
\left[\begin{array}{cc}  D^{-2}+\beta^2 D^{-2}S_k^D D^{-2} & 0 \\ 0 & D^{-2}+\beta^2 D^{-2}S_k^D D^{-2}   \end{array}\right]\left[\begin{array}{cc} A^T_W \\ A_W \end{array}\right].$$
The singular values of the similarity matrices $S_k^D=DS_kD$ are clearly changing, but the rank of the similarity matrix is unchanged and the recovery of the roles is the same as for $S_k$. 

\begin{remark}
Notice that the adjacency matrix, $A$, of a  checkerboard graph can also be considered as a weighted matrix $D|A|D$, where $|A|$ is its
(unsigned) adjacency matrix, and $D$ is the diagonal sign matrix making it signed and checkerboard. \hfill $\Box$
\end{remark}

\section{Conclusion and Future Work}
In recent years, the role extraction problem has become popular as researchers have determined a general definition of roles and have developed algorithms to find role structures within networks. In this paper, we explored analytically why a recent indirect approach using the neighborhood pattern similarity measure is able extract role structures from networks, without using any a priori knowledge of the network.

For our analysis, we first focused on an ideal graph case with a minimal role matrix and showed how the role structure can be extracted from the low-rank factorization of the similarity by clustering the rows of the low-rank factor. We then analyzed the perturbed graph case and how adding or subtracting elements in the adjacency matrix changes the singular values of the adjacency matrix and the similarity matrix.

Lastly, we unified some special complex networks structures as role structures by constructing their image matrices and showing how these matrices are minimal. From our analysis of the similarity matrix, the indirect approach using the neighborhood pattern similarity measure is able to extract these structures from the network. The unification of these structures is important because it allows us to use one approach to extract any structure without any a prior knowledge of the network. For example, community detection algorithms assumed that the network can be grouped into communities. However, for some networks, there may exist overlapping community structures, which the algorithms would fail to find. 

This paper focused on the theoretical analysis of the neighborhood pattern similarity measure with respect to the role extraction problem. A forthcoming paper will explore the efficiency of this indirect approach compared to other indirect and direct graph partitioning and role extraction  algorithms.

\section*{Acknowledgment}
Part of this work was performed while the second author was a visiting professor at UC Louvain, funded by the Science and Technology Sector, and with additional support by the Netherlands Organization for Scientific Research. 
This work was also supported by the US National Science Foundation under grants DBI 1262476 and CIBR 1934157. Wen Huang was partially supported by the Fundamental Research Funds for the Central Universities (No. 20720190060).

\bibliographystyle{imaiai}
\ifx\undefined\BySame
\newcommand{\BySame}{\leavevmode\rule[.5ex]{3em}{.5pt}\ }
\fi
\ifx\undefined\textsc
\newcommand{\textsc}[1]{{\sc #1}}
\newcommand{\emph}[1]{{\em #1\/}}
\let\tmpsmall\small
\renewcommand{\small}{\tmpsmall\sc}
\fi


\newpage

\section*{Appendix A}
Consider the formula \eqref{lambdak} for the singular values of  $S_k^\frac12$, 
\begin{equation} \label{A1}
 [\lambda_i^{(k)}]^2=\lambda_i^2\sum_{\ell =1}^k [\beta \lambda_i]^{2(\ell-1)}=\lambda_i^2\frac{(1-[\beta \lambda_i]^{2k})}{(1-[\beta \lambda_i]^{2})} ,
\end{equation}
where we used the simplified notation $\lambda_i:= \lambda_i^{(1)}$ for the singular values of  $S_1^\frac12$. 
If $\lambda_i$ and  $\lambda_j$ are two singular values of  $S_1^\frac12$  satisfying $\lambda_i > \lambda_j$, then we prove that 
$\frac{\lambda_i^{(k+1)}}{\lambda_j^{(k+1)}} >  \frac{\lambda_i^{(k)}}{\lambda_j^{(k)}},   \; \forall k \ge 1.$ Because of \eqref{A1}, we need  to show that 
$$ \frac{\sum_{\ell =1}^{k+1} [\beta \lambda_i]^{2(\ell-1)}}{\sum_{\ell =1}^{k+1} [\beta \lambda_j]^{2(\ell-1)}} =
\frac{\sum_{\ell =1}^k [\beta \lambda_i]^{2(\ell-1)}+[\beta\lambda_i]^{2k}}{\sum_{\ell =1}^k [\beta \lambda_j]^{2(\ell-1)}+[\beta\lambda_j]^{2k}}>
\frac{\sum_{\ell =1}^k [\beta \lambda_i]^{2(\ell-1)}}{\sum_{\ell =1}^k [\beta \lambda_j]^{2(\ell-1)}},
$$
and this is satisfied if and only if
$$ 
\frac{[\beta\lambda_i]^{2k}}{[\beta\lambda_j]^{2k}}>
\frac{\sum_{\ell =1}^k [\beta \lambda_i]^{2(\ell-1)}}{\sum_{\ell =1}^k [\beta \lambda_j]^{2(\ell-1)}},
$$
which is equivalent to 
$$ 
\frac{\sum_{\ell =1}^k [\beta \lambda_j]^{2(\ell-1)}}{[\beta\lambda_j]^{2k}}>
\frac{\sum_{\ell =1}^k [\beta \lambda_i]^{2(\ell-1)}}{[\beta\lambda_i]^{2k}},
$$
and to
$$ 
\sum_{\ell =1}^k [\beta \lambda_j]^{-2\ell}>
\sum_{\ell =1}^k [\beta \lambda_i]^{-2\ell}.
$$
This last inequality follows from 
$$  \sum_{\ell =1}^k [\beta \lambda_j]^{-2\ell}- \sum_{\ell =1}^k [\beta \lambda_i]^{-2\ell} \ge  [\beta \lambda_j]^{-2k}-[\beta \lambda_i]^{-2k}>0.
$$ 

In order to assess the influence of the parameter $\beta$ we look at the gap between the ratios  $\frac{\lambda_i^{(\infty)}}{\lambda_j^{(\infty)}}$ and  $\frac{\lambda_i^{(1)}}{\lambda_j^{(1)}}$. It follows from \eqref{A1} that
$$  \left[\frac{\lambda_i^{(\infty)}}{\lambda_j^{(\infty)}}\right]^2 = \left[\frac{(1-[\beta\lambda_j]^2)}{(1-[\beta\lambda_i]^2)}\right] \left[\frac{\lambda_i^{(1)}}{\lambda_j^{(1)}}\right]^2,
$$
where the scaling factor $\frac{(1-[\beta\lambda_j]^2)}{(1-[\beta\lambda_i]^2)}$ is larger than 1 if $\lambda_i>\lambda_j$ and its derivative versus $\beta$ is positive as long as $\beta\lambda_i < 1$, indicating that it grows with $\beta$.

\end{document}